%% file: main.tex
\pdfoutput=1
\documentclass{acm_proc_article-sp}
\usepackage{url}
\usepackage[backend=bibtex8]{biblatex}
\usepackage[T1]{fontenc}
\usepackage[cp1250]{inputenc}
\usepackage{rotating}
\usepackage{float}
\usepackage{clrscode3e}
\usepackage{amsmath}
\usepackage{amssymb}

\newtheorem{theorem}{Theorem}
\newdef{definition}{Definition}
\newtheorem{lemma}{Lemma}
\newtheorem{fact}{Fact}

\newcommand{\defeq}{\stackrel{def} {\equiv} }

\newcommand{\foo}[3]{
\begin{codebox}
\Procname{$\proc{#1}(#2)$}
#3
\end{codebox}
}

\newcommand{\call}[2]{\attribii{#1}{\proc{#2}}}

\newcommand{\SELECT}{\kw{SELECT} }
\newcommand{\INSERT}{\kw{INSERT} }
\newcommand{\DELETE}{\kw{DELETE} }
\newcommand{\SQLOR}{\kw{OR} }
\newcommand{\SQLAND}{\kw{AND} }

\begin{document}
  \input{title.tex}
  \input{abstract.tex}
  \input{quote.tex}
  \input{introduction.tex}
  \input{model.tex}
  \input{problem.tex}
  \input{intuitions.tex}
  \input{algorithm.tex}

  \input{proof.tex}
  \input{optimizations.tex}
  \appendix
  \input{experiment.tex}

  \printbibliography
\end{document}

%% file: title.tex
\title{Algorithm for Invalidation of Cached Results of Queries to a Single Table}
\author{
  Jakub Łopuszański \footnote{
    Department of Computer Science and Mathematics, University of Wrocław
  }
}
\date{}
\maketitle

%% file: abstract.tex
\abstract{
One of the most popular setups for a back-end of a high performance website consists of a relational database and a cache which stores results of performed queries.
Several application frameworks support caching of queries made to the database, but few of them handle cache invalidation correctly, resorting to simpler solutions such as short TTL values, or flushing the whole cache after any write to the database. In this paper a simple, correct, efficient and tested in real world application solution is presented, which allows for infinite TTL, and very fine grained cache invalidation. Algorithm is proven to be correct in a concurrent environment, both theoretically and in practice.
}

%% file: quote.tex
\begin{quote}
\emph{``There are only two hard things in Computer Science: cache invalidation and naming things.''}

\flushright{-- Phil Karlton}
\end{quote}

%% file: introduction.tex
\section{Introduction}
\label{Introduction}
As PHP, MySQL and Memcached are technologies which can be used free of charge, many startups have chosen these technologies only to find several years later that as the network traffic grows
so does the congestion at the database, unless one caches data quite aggressively.
This is probably true for any triple of a scripting language, a relational database, and a distributed hash table.
In such setups there are multiple front-end machines running an application written in a scripting language, which communicates with the back-end
which consists of one or more databases and one or more cache servers.

It is also a good idea to have cache service running on each of the front-end servers for caching data which must be available quickly without a delay introduced by a network communication.
Such a local cache is private to the machine which hosts it, which means that only this particular server can communicate with it.

Relational databases are believed to be stable, coherent, permanent storage, but slow, and not so easy to shard and scale, while caches are fast and very easy to scale due to their key-value architecture, but diskless and volatile.

As large databases tend to work slower the more traffic they have to serve, they are often replicated or split into smaller parts called \emph{shards}.
Since performing \kw{JOIN}s across sharded tables is difficult and rarely supported by databases, they are usually performed by application run on a front-end machine, which is quite more expensive, but more scalable and tractable if input and output data is cached properly.

Therefore it is a common strategy to store the result of a query in the cache, using the query as the key. This gives the best of both worlds, as one can shard cache into multiple servers while the data is safely stored in a single easy to maintain relational database.

The most difficult part of caching becomes then to know when to remove stale data from a cache, and this paper is dedicated to this problem.

If freshness and integrity of query results is not so important, then one can simply associate a Time To Live (TTL) attribute with each key-value pair stored in the cache. Memcached protocol supports this approach by supporting TTLs natively \cite{MemcachedProtocol}, so does Mysqlnd project \cite{MySQLNDQC} among many others.

Another option is to flush whole cache each time the database is modified, which while correct, seems too aggressive. Actually built-in MySQL cache uses this strategy \cite{MySQLCache}, which does not help much if a database handles several hundreds writes per second.

A correct solution which requires some extra work from a developer, but offers a reasonable balance between the two above approaches, is to manually delete only those query-result pairs from the cache which might be affected by a particular write to the database.

This task can be greatly automated, but may be ineffective if the number of queries which require invalidation is large or even unknown.
Some frameworks try to maintain an additional index which somehow connects records and cached queries together, so that invalidation can be performed automatically.
This introduces new problems, rarely solved correctly. 
One particularly wrong algorithm is to store for each record of the table an information about all cached queries which contain this record.
Although this helps handle deletion of records, it does not handle inserts well, as there is not enough information about queries which should return the newly inserted record.
This issue can be addressed by a more sophisticated data structure which holds some meta data for gaps between rows scanned during a query.
Another problem is where to store such metadata. Cache is volatile, so we risk data integrity if we decide to store this index in cache. Database is slow (which is the reason we use caching at all in the first place), so storing the mapping there actually doubles the load.

To see why this can be a problem consider a simple 3-dimensional relation $User \times Game \times Date$ which stores information about dates when users played games.
Let say the system generates following kinds of queries:

\input{queries_example.tex}

where U,G,D are parameters which can vary from query to query.
Observe that any \DELETE query of the above form must invalidate cached results to all \SELECT queries.
\begin{figure}[!htb]
\includegraphics[width=\columnwidth]{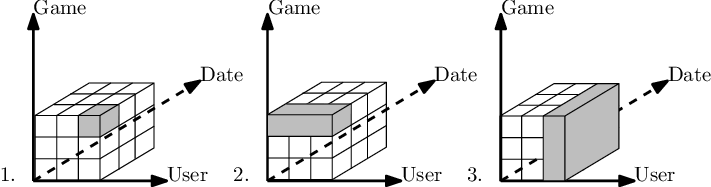}
Visualizations of subspaces for following queries
\begin{enumerate}
\item \INSERT \kw{INTO} \id{Played} (\id{User},\id{Game},\id{Date}) \\ \kw{VALUES} (2,2,0)
\item \SELECT \kw{COUNT}(*) \kw{FROM} \id{Played} \\ \kw{WHERE} \id{Game} = 2 \kw{AND} \id{Date} = 0
\item \DELETE \kw{FROM} \id{Played} \\ \kw{WHERE} \id{User} = 2
\end{enumerate}
\label{3DExample}
\end{figure}
Geometrically this is because the 1-dimensional space
scanned by \SELECT query intersects with the 2-dimensional subspace deleted by \DELETE query regardless of particular parameters U,G,D.
Examples of subspaces corresponding to the three types of queries are depicted in Figure \ref{3DExample}.

A solution which would try to keep track of all the dependencies between queries explicitly and invalidate cached results of every \SELECT one by one after every \DELETE would have to take a lot of time. On the other hand neglecting to do so immediately, may result in serving stale data from cache. Also the dependency graph can easily become larger than the original table itself. Ideas in which dependencies are tracked at record level may need even more memory. 

To solve the problem correctly and efficiently, another solution called \emph{generational keys} can be employed.
It seems to be a part of the folklore and it is hard to track back the origins of it -- see \cite{GenerationalCache},\cite{GenerationalCache2},\cite{GenerationalCache3} for example.
The key idea is that we think about the keys which need to be invalidated together as generations.
Whenever the data in the database changes in a way which should invalidate whole generation, we simply increment the number of the current generation called \emph{revision number}.
Each time we read the data from cache, we also fetch the revision number, and if the data stored in cache belongs to the old generation, we simply ignore it.
This technique allows invalidating multiple keys in the cache using a single increment operation.

The original formulation of this technique made the revision number a part of the key, which had some drawbacks. 
Fetching any data had to be done in two round trips : first to get the revision number, second to get the actual data.
Another problem was a huge cache pollution caused by many no longer used key-value pairs from old generations \cite{PeepMemcached}.
Theoretically this should not be an issue for an LRU cache implementation, but in practice slab memory allocator and lazy garbage collector used in Memcached were affected by this.
Here we propose and use a different approach in which the revision number is stored in the value, not in the key, which solves both problems.

This technique can be used as follows for our simplistic example of \id{Played} table. We store the revision number in a cache accessible from all front-end machines. Each time the \DELETE is performed, the revision number is incremented. Cached results of \SELECT queries could be then versioned, by appending the current revision number to the cached result. Each time front-end server fetches the cached result from cache it must fetch the revision number as well and reject the cached result if the revision number does not match. If the revision itself is missing we reset it to the current timestamp times 1000 (or other value guaranteed to be larger than any previous).

This is a very simplistic example, which quickly becomes more and more complicated as we introduce different queries
to the system. For example, to properly handle \INSERT queries in our example, we also need a separate revision number for each pair in $Game \times Date$. The
developer has to remember that after inserting (U,G,D) she needs to increment a revision number for particular pair $(G,D)$. The routine for selects must also be adapted to check for this additional revision number. 
Keeping track of revision numbers which need to get incremented and verified can become a maintenance nightmare and source of hard to reproduce bugs, when left to humans.

This paper proposes a way to formalize and automate this technique. The solution presented in this paper is:
\begin{description}
\item [correct] -- data returned by selects are never significantly older then the moment when the select was called
\item [fast] -- the number of accessed keys in a cache during a query is a constant dependent only on the number of columns in the table
\item [fine grained] -- only the results which can not be proven fresh are invalidated
\item [practical] -- it was successfully used in two web applications with hundreds of thousands of users
\item [fully automatic] -- relieves developers from bookkeeping dependencies between queries even in the presence of arbitrary large number of different queries (for example when different \kw{ORDER}s, \kw{LIMIT}s or \kw{OFFSET}s are used) 
\end{description}
The main result of this paper is that for a $k$-dimensional relation, we need to fetch at most $1+2^k$ keys from a cache in order to get the cached result of a query, and that after each write to the table we need to increment only $2^k$ keys to invalidate old results.
Even though this number is already small, these can be implemented as a single round trip to the cache server, if it supports bulk queries, as names of all keys which are referenced by the algorithm are known up front.
\subsection{Outline}
The rest of this paper is organized as follows. Section \ref{Model} describes the model of a database and caches used throughout this paper.
Section \ref{Intuitions} is intended to provide some insights into the design of the algorithm and intuitions behind it.
Section \ref{Algorithm} contains a description of the algorithm.
Section \ref{Proof} presents a proof of correctness for this algorithm.
Section \ref{Optimizations} sketches several optimizations and possible extensions to the algorithm.

%% file: queries_example.tex
\begin{enumerate}
\item \kw{INSERT} \kw{INTO} \id{Played} (\id{User},\id{Game},\id{Date}) \\ \kw{VALUES} (:U,:G,:D)
\item \kw{SELECT} \kw{COUNT}(*) \kw{FROM} \id{Played} \\ \kw{WHERE} \id{Game} = :G \kw{AND} \id{Date} = :D
\item \kw{DELETE} \kw{FROM} \id{Played} \\ \kw{WHERE} \id{User} = :U
\end{enumerate}

%% file: model.tex
\section{The Model}
\label{Model}

Although the algorithm was successfully implemented for PHP, MySQL and Memcached it is quite general and can be used for other, even no-SQL, databases, and different DHT implementations. Therefore, let us define some abstractions of a database, a query, and a cache.
The algorithm will handle caching of queries to a single table with $k$ columns of arbitrary type.
In practice, because of sharding and caching, \kw{JOIN}s are often performed at front-end machines, and not by the database itself, so the limitation of our considerations to a single table is justified.
Perhaps, one might apply this algorithm to a result of joining a few tables together, but this would certainly require some extra conceptual work.
Let us fix a $k$-dimensional space $S$, where each dimension corresponds to a single column of the only table in our database.
We do not require a type of a column to be numeric, as it is enough for us to be able to serialize its values to a string, and test them for equality.

\begin{definition}
A \emph{record} is a $k$-dimensional vector in $S$. If $r$ is a record then $r[i]$ for $i=1,\ldots,k$ is its $i$-th field.
\end{definition}

\begin{definition}
A \emph{table} is a finite set of records. 
\end{definition}

\begin{definition}
A \emph{query} is a record in which special wildcard placeholders * can occur.
\end{definition}

\begin{definition}
A \emph{subspace of a query} $q$, is defined as 
$$subspace(q) = \{ v \in S \: | \: \forall i . \: v[i]=q[i] \lor *=q[i]\} \enspace .$$ 
\end{definition}

\begin{definition}
  A \emph{database} is data structure which contains a \id{table} and exposes interface functionally equivalent to the pseudocode specification below:
  \input{database_spec_src.tex}
\end{definition}

\begin{definition}
  A \emph{cache} is a data structure which contains a key-value \id{mapping} and exposes interface functionally equivalent to the pseudocode specification below:
  \input{cache_spec_src.tex}
\end{definition}

\begin{definition}
A call to \proc{add}, \proc{set} or \proc{increment} methods of a cache can be \emph{successful} or not.
An \proc{add} is successful iff it returned \const{true}, an \proc{increment} is not successful iff it returned \const{undefined} and \proc{set} is always successful.
\end{definition}

\begin{definition}
An \emph{eviction} of a key from a cache is a situation in which a cache server run out of storage and had to delete the key from its memory to make room for new data.
\end{definition}
\begin{definition}
A \emph{horizon} of a cache is a lower bound for the time elapsed between putting a key in the cache and the moment it gets evicted from it. 
\end{definition}

In practice for LRU caches horizon greatly depends on the length of the LRU queue and the frequency of unique writes to it, and is rarely smaller than several hours.

In order to model evictions, we allow a cache to spontaneously call $\proc{delete}(key)$ for any $key$ at arbitrary chosen time, but not before the thread which called $\proc{set}(key,\ldots)$ or $\proc{add}(key,\ldots)$ finished executing our algorithm. As serving a single user rarely takes more than half a second, and horizon tends to be measured in hours, this model is quite realistic. It allows us to focus on concurrency issues more than on a reliability of a cache as a storage.

The above definition of a query does not reflect the whole potential of SQL, but is enough to model many CRUD and ORM systems.
The vocabulary consisting of \proc{select}, \proc{delete}, and \proc{insert} is quite restricted to make the presentation of the algorithm simpler, but can be extended to handle keywords such as \kw{MAX}, \kw{COUNT}, \kw{LIMIT}, \kw{OFFSET} or \kw{UPDATE} with a little extra effort.

In particular this definition allows only equality constraints in the select query, but the idea can be applied to queries containing other constraints by virtually rewriting them into a two stage queries : first we use equality constraints (if any) to limit the resulting set, then we further filter it by other constrains. This rewriting operation is just for the purpose of analysis and does not have to be implemented.
That is if an application performs\\
\SELECT \kw{COUNT}(*) \kw{AS} \id{cnt} \\
\kw{FROM} \id{Played} \\
\kw{WHERE} \id{Date}>123456 \\
\kw{AND} \id{User}=2\\
\kw{GROUP BY} \id{Game}\\
\kw{ORDER BY} \id{cnt}\\
then for the purpose of the analysis we will model that as $\proc{select}((2,*,*))$ as the only thing that is important for us is
the scope scanned by the query, and this does not depend on \kw{GROUP}ing, \kw{ORDER}ing, nor \kw{COUNT}ing.
It does however depend on the inequality $\id{Date}>123456$ but our simple algorithm will not be able to take any advantage from this constraint.

Similarly, one can often mentally emulate \kw{UPDATE} with \SELECT followed by \DELETE and \INSERT.
In general to apply the algorithm from this paper, for each query we need to know what is the smallest subspace containing all records it reads and what is the smallest subspace which contains all the records it deletes or creates.
That is, if one thinks about records in terms of points in a space, we need a bounding box. 
Intuitively if a bounding box of a read query intersects a bounding box of a write query, then the later one should invalidate cached results of the first one.
The smaller the box, the less it interferes with other queries, so we can use the tightest upper bound we can prove for a particular query.
Our algorithm can not infer any additional knowledge from the condition $age>21$ so it simply ignores it, and does not narrow the subspace along the age axis at all.

We assume a database and two caches, called local and global, to be available from each front-end machine which executes the algorithm.
In some applications local and global might be different names of the same cache, but having a separate instance of a local cache at each of the front-end machines
reduces the problem with network latency and congestion. The drawback is that each of these instances is private to a front-end machine which hosts it and can not be accessed from other machines, which imposes some difficulties with cache invalidation. For example if front-end node A performs a delete, then a front-end node B is not aware of this change. Therefore local cache is a good place to store information which does not change in time, but quite risky for things which change a lot. Our algorithm will take advantage of local caches, but can be used as well in environments in which they are not available by simply using the global cache in place of the local cache.

%% file: database_spec_src.tex
\foo{select}{query}{
\zi \Return $\id{table} \cap \func{subspace}(\id{query})$
}
\foo{delete}{query}{
\zi $\id{table} \leftarrow \id{table} \setminus\, \func{subspace}(\id{query})$
}
\foo{insert}{record}{
\zi $\id{table} \leftarrow \id{table} \cup \{\id{record}\}$
}

%% file: cache_spec_src.tex
\foo{get}{key}{
\zi \Comment key can be of any type which can be serialized to a string
\zi \If {$\id{key} \in \id{mapping}$}
\zi \Then
     \Return $\id{mapping}[\id{key}]$
\zi \Else
\zi    \Comment a special constant indicating a miss
\zi    \Return $\const{undefined}$
    \End
}
\foo{multiget}{keys}{
\zi \Comment this is expected to be faster than separate calls to get, 
\zi \Comment but does not have to be atomic
\zi \Return $\proc{map} \enspace \proc{get} \id{keys}$
}
\foo{set}{key,value}{
\zi \Comment value can be of any type which can be serialized to a string
\zi $\id{mapping}[key] \leftarrow \id{value}$
}
\foo{add}{key,value}{
\zi \Comment this function must be atomic
\zi \If {$\id{key} \in \id{mapping}$}
\zi \Then
     \Return $\const{false}$
\zi \Else
\zi    $\proc{set}(\id{key},\id{value})$
\zi    \Return $\const{true}$
    \End
}
\foo{increment}{key}{
\zi \Comment this function must be atomic
\zi \If {$\id{key} \in \id{mapping}$}
\zi \Then
     $\id{value} \leftarrow \proc{get}(\id{key})+1$
\zi   $\proc{set}(\id{key},\id{value})$
\zi   \Return $\id{value}$
\zi \Else
\zi   \Return $\const{undefined}$
    \End
}
\foo{delete}{key}{
\zi $\id{mapping}[key] \leftarrow \const{undefined}$
}

%% file: problem.tex
\section{Problem Statement}
\label{Problem}
The problem is to design a data structure which uses the original database and two caches : local and global, to implement the same interface as the database.
We require the new data structure to return fresh results, which means that \SELECT operation performed at moment $t$ should provide the result which the original database would give at some moment $t' \ge t-\epsilon$ for the same query. 
Here $\epsilon$ is a small constant, which is an upper bound on the execution time of the algorithm (think: milliseconds).
The goal is to minimize the number of queries to the original database, and number of queries and round trips required for communication with caches.

%% file: intuitions.tex
\section{Intuitions}
\label{Intuitions}

\begin{figure}[!htb]
\includegraphics[width=\columnwidth]{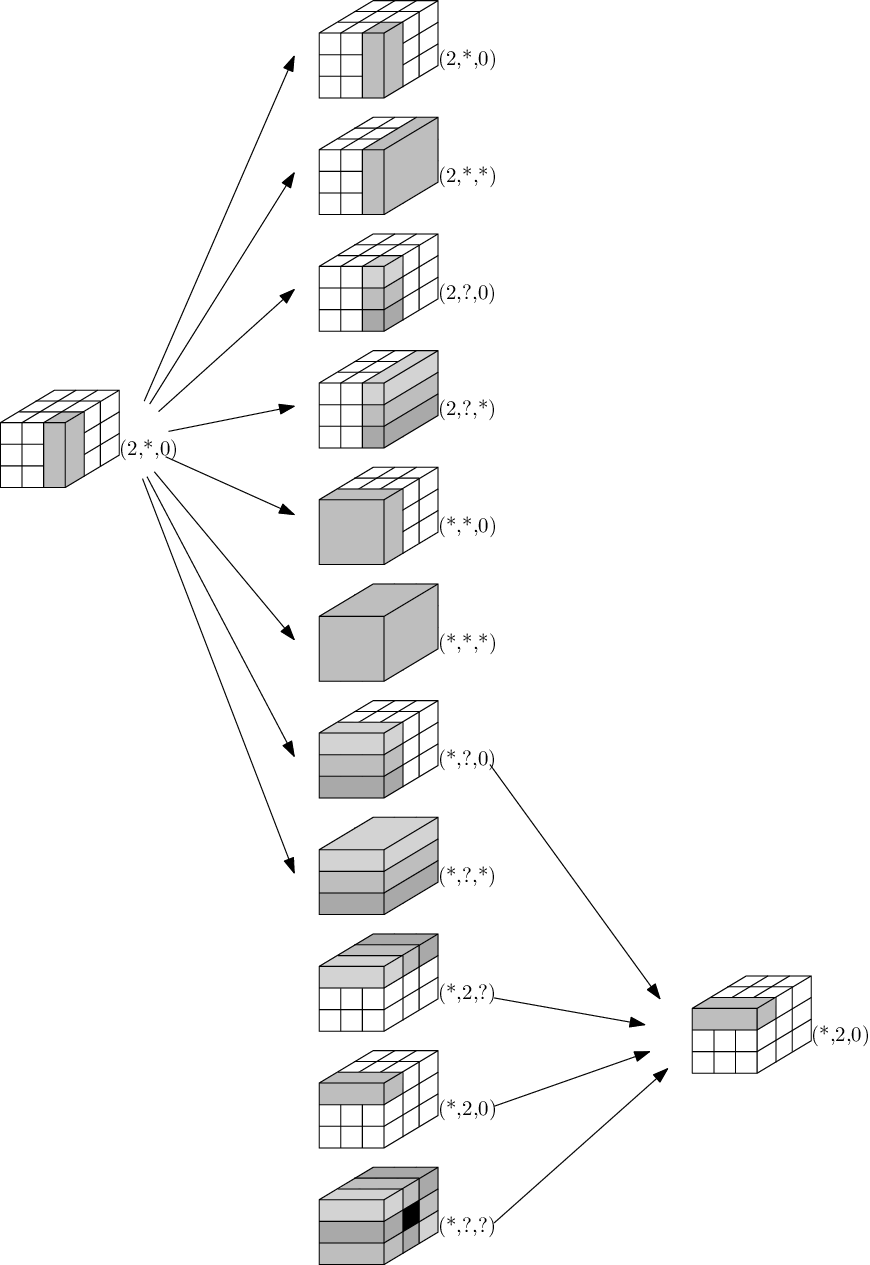}
On the left we see a node which represents a subspace of query \mbox{\DELETE \kw{FROM} xyz \kw{WHERE} x=2 \kw{AND} z=0}.
On the right we see a node corresponding to a query \mbox{\SELECT * \kw{FROM} xyz \kw{WHERE} y=2 \kw{AND} z=0}.
In the middle we see all nodes which get invalidated after the \DELETE, as well as all nodes which are checked before the \SELECT.
Since subspaces of the two queries intersect, so do their neighborhoods in the graph.
\label{Boxes}
\end{figure}

If you think about records as points in a $k$-dimensional space, and see queries as subspaces, then it is easy to see that
if the result of a query $A$ depends only on $subspace(A)$ and a query $B$ adds or removes points only within $subspace(B)$ which is disjoint from $subspace(A)$, then $B$ has no influence on the results of $A$.

Therefore, we will be on the safe side if after a write query $X$ which affects $subspace(X)$ we will invalidate cached results of all queries $Y$,
such that $subspace(Y)$ and $subspace(X)$ intersect.
The problem is to do it quickly, correctly and without tons of memory to keep track dependencies.

To see where the difficulty lays imagine a directed bipartite graph in which both layers contain one node for each possible query.
Left side represents write queries, and right side represents read queries.
Edges in the graph will represent dependence -- there is an edge from $X$ to $Y$, iff $subspace(X) \cap subspace(Y) \neq \emptyset$.
Observe several important facts:
\begin{itemize}
  \item number of nodes (possible queries) is infinite if at least one of the dimensions is infinite 
  \item number of edges is even larger than the number of nodes
  \item number of edges outgoing from a single write query may be infinite
  \item number of edges incoming to a single read query may be infinite
\end{itemize}
The last two points shed some light on the issues which can be faced when invalidating cached results after a single write query, or when checking freshness of cached results during a read query.
It is not to say that a system would have to track infinite amount of data, but, non the less, it seems to be bounded only by the number of different queries which the system is able to generate, and the bound is not even linear.

The proposed algorithm will add a third, intermediate layer to the graph. 
Roughly speaking, we want to factorize the dense dependency matrix of the original bipartite graph into two sparser matrices. 
This third layer contains all possible queries, as well as all vectors which resemble queries with $?$ placeholders in arbitrary places.
Figure \ref{Boxes} depicts a small portion of such a tripartite graph for $k=3$ containing neighborhood of two vertices from left and right layer, in order to demonstrate how the middle layer provides a bridge between them.
These additional nodes will serve as a junction points between a write query and multiple read queries.
To be more precise, a node in the middle layer will be connected to all nodes in the right layer which represent queries resembling the center node except for a few numbers in places of ? placeholders. For example a middle node $(*,?,1)$ will be connected to $(*,0,1),(*,1,1),(*,2,1)\ldots$ on the right.
The intuition behind this is that often we want to invalidate all queries which differ only by a parameter at a particular position.
In the example mentioned in Section \ref{Introduction} the \DELETE query $(U,*,*)$ should invalidate results of all \SELECT queries of the form $(*,?,?)$.
Of course there are other results which also need to be invalidated, which will be represented in our new graph by edges from left layer to the middle layer.
More precisely a node in the left layer will be connected to all nodes in the middle layer which resemble it except for a few * in places of non-* values and
a few ? in places of *.
So for example a query $(U,*,*)$ from the left layer should be connected to the nodes\\\mbox{(U,*,*),(*,*,*),(U,?,*),(*,?,*),(U,*,?),(*,*,?),(U,?,?),(*,?,?)}.\\
The intuition behind this is that the subspace of a write query intersects with another subspace if and only if they agree on all positions without *.
So if $i$-th coordinate of a write query is a *, then a dependent read query can contain a star or any non-* value, which we model using a ? placeholder.
If, $i$-th coordinate of a write query is a non-*, then the a dependent read query must either contain the very same value or a star at this position.
As we will see, the transitive closure of the new graph is exactly equal to the original set of edges in the bipartite graph.
While the graph is arguably larger, it has some nice features:
\begin{itemize}
\item there is exactly $2^k$ edges outgoing from any write query
\item there is at most $2^k$ edges incoming to any read query
\end{itemize}

If we associate an integer counter with each node of the middle layer, and increment it each time an incident query on the left modifies the database,
then it is enough to check if counters associated to nodes incident to the read query have not changed since the time we cached the result to know if we can use the cached result or not.

Actual algorithm has to be a little bit more complicated to handle cache misses caused by evictions in a correct way.
Some counters can be missing and we need to reinitialize them, carefully choosing a value, which must be larger then any previous value.

%% file: algorithm.tex
\section{The Algorithm}
\label{Algorithm}
We will present the algorithm as a wrapper around database and two caches, which itself implements the interface of a database.
As noted before this pseudo code is a simplification which does not deal with SQL parsing etc.
Actually there are many ORM frameworks which hide SQL manipulation from developer and adding the of our algorithm to them should be even easier.
Also the algorithm assumes that all $k$ dimensions (columns) are relevant.
Since the complexity depends on $k$ one can use some domain knowledge to limit the number of relevant columns only to those which are used in equality constraints.

\input{algorithm_src.tex}

%% file: algorithm_src.tex
Auxiliary function \proc{allVariantsOf} takes a query and substitution \id{rules} and returns all possible vectors that can be obtained by using
these rules zero or more times. 
For example (*,3,2) with rules $\{nonStar \to \enspace ?\}$ should return exactly four vectors: $[(*,3,2),(*,?,2),(*,3,?),(*,?,?)]$
\foo{allVariantsOf}{query,len,rules}{
\zi  \If {$len=0$}
\zi  \Then 
        \Return $[[\;]]$
     \End
\zi  $prefixes \leftarrow \proc{allVariantsOf}(\id{query},\id{len}-1,\id{rules})$
\zi  $extended \leftarrow \proc{map} ((\proc{append} \; \id{query}[\id{len}]), \id{prefixes})$
\zi  \If {$query[len] \notin \id{rules}$}
\zi  \Then
        \Return $extended$
     \End
\zi  $\id{alternative} \leftarrow \proc{map} ((\proc{append} \; \id{rules}[\id{query}[\id{len}]]),\id{prefixes})$
\zi  \Return $\proc{concatenate}(\id{extended},\id{alternative})$
}

The \proc{getRevisions} function returns revisions of \id{subspaces}.
It has to deal with occasional cache misses, which it fixes by trying to reset revision to a \id{monotoneValue} which is guaranteed to be larger
than the latest revision for this subspace before it was evicted from cache.
Different threads can have slightly different clock settings, and execute scripts at different speeds, so we can not assume that a \id{monotoneValue}s
multiple threads are trying to put in the \id{globalCache} are equal.
The only requirement we need here is that each thread computed a value larger than the latest value of revision just before it got evicted from the cache.
This can be done by adjusting \const{maxQueriesPerTimeStep} and assuming that horizon of the cache is large enough.

\foo{getRevisions}{subspaces}{
\zi $\id{monotoneValue} \leftarrow \proc{now}()*\const{maxQueriesPerTimeStep}$
\zi $\id{revisions} \leftarrow \id{globalCache}.\proc{multiget}(\id{subspaces})$
\zi \For {$i \in \{i | \id{revisions}[\id{i}]=\const{undefined} \}$} 
\zi \Do
     \If {$\id{globalCache}.\proc{add}(\id{subspaces}[\id{i}],\id{monotoneValue})$}
\zi   \Then
        $\id{revisions}[\id{i}] \leftarrow \id{monotoneValue}$
      \End
    \End
\zi  $\id{unknown} \leftarrow \{ \id{subspaces}[\id{i}] | \id{revisions}[\id{i}]=\const{undefined} \}$
\zi \If {$|\id{unknown}| >0$}
\zi \Then
      $missing \leftarrow \proc{getRevisions}(\id{unknown})$
\zi    $\id{revisions} \leftarrow \proc{merge}(\id{revisions}, \id{missing})$ 
    \End
\zi \Return \id{revisions}
}
A missing revision's value is added using \proc{add}.
This is important under race conditions.
Suppose there are four threads A,B,C,D.
If \proc{set} was used instead of \proc{add} then it would be possible that thread A set revision to $r$, thread B incremented it to $r+1$, then thread C restored it back to $r$, which would violate the monotonicity property of revisions and could result in thread D reading stale data, which thread B intended to invalidate.

Another important issue is that in case of \proc{add} failure we can not use our copy of \id{monotoneValue}, but rather should fetch the value added by another thread.
This is done by a recursive call, which given the assumption about long horizon, should finish successfully without further recursion.
If thread A used its \id{monotoneValue} in this case, it could happen that it was larger than the value added to the cache by thread B, and thread A would then store cached results of database query tagged by revision number which is too large. Imagine that much later, thread C, after a write to the database, increments the revision which now becomes equal to \id{monotoneValue} used long time ago by A. This could result in thread D reading stale data stored by A.

Please note that \proc{allVariantsOf} is not atomic, thus the returned revision numbers perhaps never coexisted in the cache at the same point in time.
However, each of these numbers is not smaller than maximum real value of revision seen up to the moment of the call to this function, 
and, moreover, existed in the cache at some point in time.
This is important as it implies that algorithm will refuse to use cached data which was invalidated before the call to this function,
and that it will not store cached results tagged as a version that was not yet reached.

The \proc{select} function presented below first gathers information about current version of the subspaces intersected by the query,
and then fetches the cached result. 
In case of a cache miss or mismatched versions, it forwards the call to the database.
Note that short numeric versions are stored in the \id{globalCache} while possibly large result is stored in the \id{localCache}.
Although our simplistic model results in a one-to-one correspondence between queries and their subspaces, 
the algorithm uses a \proc{digest} function such as \func{sha1} to convert the \id{query} to a unique string which serves as a \id{key} in the \id{localCache}.
This is to demonstrate how to handle more realistic situations where two different queries can have same subspace without causing
collisions in \id{localCache}. For example we could have two queries asking for different columns, sort order, limit, offset or aggregate function, but scanning the same subspace because of identical \kw{WHERE} clauses. This would be handled correctly by having \proc{digest} return two different \id{key}s for them, even though \proc{allVariantsOf} returns the same set of \id{variants}.

\foo{select}{query}{
\zi $\id{variants} \leftarrow \proc{allVariantsOf}(\id{query},\const{k},\{nonStar \to\enspace ?\})$
\zi $\id{key} \leftarrow \proc{digest}(\id{query})$
\zi $\id{revisions} \leftarrow \proc{getRevisions}(\id{variants})$
\zi $\id{version} \leftarrow \proc{join}(\id{revisions},``.")$
\zi $\id{cached} \leftarrow \call{localCache}{get}(\id{key})$
\zi \If {$\id{cached} = \const{undefined} \lor \attribii{cached}{version} \not \succeq \id{version}$}
    \Then
\zi    $\id{cached} \leftarrow \attribii{globalCache}{get}(\id{key})$
\zi    \If {$\id{cached}=\const{undefined} \lor \attribii{cached}{version} \not \succeq \id{version}$}
       \Then
\zi       $\attribii{cached}{result} \leftarrow \call{database}{select}(\id{query})$
\zi       $\attribii{cached}{version} \leftarrow \id{version}$
\zi       $\call{globalCache}{set}(\id{key},\id{cached})$
       \End
\zi    $\call{localCache}{set}(\id{key},\id{cached})$    
    \End
\zi \Return $\attribii{cached}{result}$
}
The concatenation operation \proc{join} and partial order $\succeq$ are defined so that 
$$\proc{join}(a,``.") \succeq \proc{join}(b,``.") \iff \forall i \, . \, a[i] \ge b[i]\enspace,$$
which means that one version is newer than the other. Note that $version$ computed here can in some rare circumstances do not reflect any particular moment in time, due to the non-atomicity of multiget, but the algorithm correctly deals with this issue. 
One can replace $\not \succeq$ with a simple $\neq$ without deteriorating algorithm's performance significantly, if this somehow seems more secure or easier to implement.

The \proc{invalidate} function invalidates all subspaces which intersect with subspace of a given query.
Doing so explicitly could require incrementing infinitely many revision numbers, 
as in place of a star in $\func{subspace}(\id{query})$ we should try every possible value for that column.
To reduce the number of steps to $2^k$ we use a special placeholder ? which is intended to have a meaning of \emph{any particular value}.
Therefore $\proc{invalidate}((*,2,3))$ will need to increment exactly 8 revisions:\\ 
\mbox{(*,2,3),(?,2,3),(*,*,3),(?,*,3),(*,2,3),(?,2,*),(*,*,*),(?,*,*)}

\foo{invalidate}{query}{
\zi $\id{rules} \leftarrow \{* \to\enspace ?, \id{nonStar} \to *  \}$
\zi $\id{variants} \leftarrow \proc{allVariantsOf}(query,\const{k},\id{rules})$
\zi \For {$i \in 1,\ldots,|variants|$}
\zi \Do
      $\call{globalCache}{increment}(variant)$
    \End
}

\foo{delete}{query}{
\zi $\call{database}{delete}(query)$
\zi $invalidate(query)$
}

\foo{insert}{record}{
\zi $\call{database}{insert}(record)$
\zi $invalidate(record)$
}

%% file: proof.tex
\section{The Analysis}
\label{Proof}
\begin{definition}
\label{EDefinition}
The \emph{query dependency graph} is a bipartite graph $<(Q,Q),E>$, where $Q$ is the set of all possible queries, and $E=\{(u,v) | \func{subspace}(u) \cap \func{subspace}(v) \neq \emptyset \}$ 
\end{definition}

\begin{definition}
\label{EEDefinition}
The \emph{revision dependency graph} is a tripartite graph $<(Q,P,Q),(E',E'')>$, where $Q$ is the set of all possible queries, $P=\{ p | \exists q \in Q . \forall i . q[i]=p[i] \lor p[i]=?  \}$ is a set of all possible queries with some of coefficients replaced with question marks, 
and $E' \subset Q \times P$ and $E'' \subset P \times Q$ are defined as follows:
$$E' = \{ (q,p) | \forall i. 
  q[i]=p[i] 
  \lor 
  (q[i]\neq * \land p[i]=*) 
  \lor
  (q[i]=* \land p[i]=?)
\}$$
$$E'' = \{ (p,q) | \forall i. 
  p[i]=q[i]
  \lor 
  (p[i]=? \land q[i]\neq *)
\}$$
\end{definition}

\begin{lemma}
\label{IntersectWhenAgree}
$$ (q,q') \in E 
\iff 
  \func{subspace}(q) \cap \func{subspace}(q') \neq \emptyset 
$$
$$
\iff
  \forall i .\enspace q[i]=q'[i] \lor q[i]=* \lor q'[i]=*
  $$
that is subspaces are disjoint if at some position $i$ they have different non-star values and that otherwise they do intersect.
\end{lemma}
\begin{proof}
This can be proven by constructing a witness $r$ which belongs to the intersection --  
$$r[i]= \left \{
\begin{tabular}{ll}
$q[i]$ & if $q[i]\neq *$ \\
$q'[i]$ & if $q[i]= * \land q'[i]\neq *$ \\
anything & otherwise
\end{tabular}   
\right .
$$
by \emph{anything} we mean here any valid value for this particular dimension, i.e. minimal possible value to avoid axiom of choice.
\end{proof}

\begin{lemma}
\label{TransitiveClosure}
The transitive closure of the revision dependency graph gives the query dependency graph, that is:
$$
(q,q') \in E \iff \exists p \in P .\enspace (q,p) \in E' \land (p,q') \in E''$$
\end{lemma}
\begin{proof}
$(=>)$ We can construct a tuple $p$:
$$p[i]= \left \{
\begin{tabular}{ll}
$q[i]$ & if $q[i]=q'[i]$ \\
$?$ & if $q[i]= * \land q'[i]\neq *$ \\
$*$ & if $q[i]\neq * \land q'[i]=*$
\end{tabular}   
\right .
$$
and verify using Definition \ref{EEDefinition} that $p$ is connected to $q$ and $q'$ by case inspection.

$(<=)$ we can use $p$ to show for each $i$ that it must be one of the following cases:
\begin{itemize}
\item[Case 1.] $p[i]=?$. From the definition of E' it must be that $q[i]=*$. 
\item[Case 2.] $p[i]=*$. From the definition of E'' it must be that $q'[i]=*$.
\item[Case 3.] $p[i]\notin\{?,*\}$. From the definitions of E' and E'' it must be that $p[i]=q[i]=q'[i]$.
\end{itemize}
In none of these cases it is possible for $q[i]$ and $q'[i]$ to be different non-star values, so applying Lemma \ref{IntersectWhenAgree} gives the thesis.
\end{proof}

\begin{fact}
For every $q,q'$ if $\enspace\func{subspace}(q) \cap \func{subspace}(q') = \emptyset$ then result returned by a database for \proc{select}$(q)$ stays the same after executing \proc{delete}$(q')$ on it.
\end{fact}

\begin{fact}
For every $q,r$ if $r \notin \func{subspace}(q)$ then result returned by a database for \proc{select}$(q)$ stays the same after executing \proc{insert}$(r)$ on it.
\end{fact}

\begin{fact}
For every $q,q'$ result returned by a database for \proc{select}$(q)$ stays the same after executing \proc{select}$(q')$ on it.
\end{fact}

\begin{lemma}
The \proc{getRevisions} performs no more than one recursive call.
\end{lemma}

\begin{proof}
The recursive call occurs only when \id{unknown} is non empty.
All the keys in variable \id{unknown} were missing during \call{globalCache}{multiGet} and were added shortly afterwards by some other thread, as \call{globalCache}{add} for them failed.
Nowhere in the proposed algorithm we delete any keys.
The only reason a key can be missing after it is \proc{add}ed by the algorithm, is an eviction.
We assumed however, that horizon is long enough, so that eviction of a key can not happen until the current thread finishes.
Therefore during the recursive call all of the keys from \id{unknown} will be still in cache, and \id{globalCache}.\proc{multiGet} will return all of them. 
\end{proof}

\begin{fact}
The proposed algorithm for \proc{select} fetches at most $2^{k+1}+1$ keys from \id{globalCache}, 1 key from \id{localCache}, performs at most $2^k$ \proc{add}s to the \id{globalCache}, and at most 1 \proc{set} to \id{globalCache} and \id{localCache}.
\end{fact}

\begin{fact}
Whole communication during \proc{select} can be performed in a constant number of round trips if \id{globalCache} supports multiple \proc{add}s in a single bulk request. Otherwise, the algorithm may require up to $O(2^k)$ round trips in the worst case.
\end{fact}

\begin{fact}
The result returned by proposed \proc{select}$(q)$ algorithm when called at moment $t$ is equal to the result that was returned by \id{database}.\proc{select}$(q)$ at some moment $t'$.
\end{fact}

\begin{definition}
Let \func{realRevision}$(p,t)$ be the maximum over all values successfully assigned trough \proc{add}, \proc{set} or \proc{increment} to the key $p$ before moment $t$.
The moment of an assignment is determined by cache server and may be a little later than the moment of calling the method and a little sooner than returning from it.
If key $p$ was never successfully assigned a value, we assume $\func{realRevision}(p,t)=0$.
\end{definition}

\begin{definition}
Let \func{realVersion}$(q,t)$ be a value computed in a same manner as the variable \id{version} in \proc{select}$(q)$, but for each key $p$ using \id{realRevision}$(p,t)$ in place of a value returned by \id{globalCache}.\proc{multiget}. That is   
$$\func{realVersion}(q,t) \defeq \proc{JOIN}( \{\id{realRevision}(p,t) | (p,q)\in E'' \} , ``.")$$
\end{definition}

\begin{fact}
\label{RevisionIsMonotone}
If a moment $t'$ occurs after a moment $t$, then 
$$\forall p \in P .\quad \id{realRevision}(p,t') \ge \id{realRevision}(p,t)\enspace.$$
\end{fact}

\begin{fact}
\label{VersionIsMonotone}
If a moment $t'$ occurs after a moment $t$, then 
$$\forall p \in P .\quad \id{realVersion}(q,t') \succeq \id{realVersion}(q,t)\enspace.$$
\end{fact}

\begin{lemma}
\label{WriteStrictlyRises}
For any $p \in P$ , if at a moment $t$ we successfully call \proc{add} or \proc{increment} for this particular $p$, 
then the newly assigned value is strictly larger than any before. That is, for any moment $t'$ before the call, and $t''$ after returning from it we have
$\func{realRevision}(p,t') < \func{realRevision}(p,t'')$.
\end{lemma}
\begin{proof}
This can be shown by induction over successful write operations in chronological order for a fixed key $p$.
The algorithm performs only \proc{add} and \proc{increment} write operations on this key.
An \proc{increment} on a missing key always fails, so the first successful operation (if any) had to be \proc{add}, and since the assigned value is strictly larger than zero, the basis of the induction holds.
Now, assume that the last successful write operation before a moment $t$ assigned the highest value so far $r$, and that at the moment $t$ we perform another successful write operation, which can be either
\begin{itemize}
\item \proc{add}, which succeeds only if key $p$ was missing. As we assume the horizon to be large enough, and \const{maxQueriesPerTimeStamp} to be chosen correctly, we can easily show that the new value is larger than $r$ even if clocks of machines are a little bit desynchronized,
\item or \proc{increment}, which succeeds only if key $p$ is still not missing, and therefore the value gets changed from $r$ to $r+1$.
\end{itemize}
In both cases the new value is larger than previous.
\end{proof}

\begin{lemma}
\label{HitIsReal}
At any moment $t$ the
$\attribii{globalCache}{mapping}[p]$ is equal to either \const{undefined} or $\id{realRevision}(p,t)$, that is a 
\proc{get} or \proc{multiget} either returns the real revision, or signals a cache miss.
\end{lemma}
\begin{proof}
The only difficulty in this Lemma is that we defined \id{realRevision} to be the maximum over all successful write operations, while the Lemma
states something about the current value.
From Lemma \ref{WriteStrictlyRises} we know that each new value is actually the largest so far.
\end{proof}

\begin{lemma}
\label{VersionIsFresh}
If \proc{select}$(q)$ is called at a moment $t$ then the value assigned to the variable \id{version} is $\succeq \func{realVersion}(q,t)$.
\end{lemma}
\begin{proof}
For each subspace $p$ in the variable \id{variants} the value returned by \proc{getRevisions} and assigned to the variable \id{revisions} come either from \id{globalCache}.\proc{multiget} or was used in a successful \id{globalCache}.\proc{add} for that key $p$.
If it was from \proc{multiget} then from Lemma \ref{HitIsReal} and Fact \ref{RevisionIsMonotone} it follows that this value was not smaller than \func{realRevision}$(p,t)$.
If it was used in a successful \id{globalCache}.\proc{add} then it hat to be at the moment $t' > t$, and it has to be equal to \func{realRevision}$(p,t')$ which according to Fact \ref{RevisionIsMonotone} is not smaller than \func{realRevision}$(p,t)$.
The $\succeq$ relation was defined so that \proc{join}ing greater or equal values together yields a greater or equal result.
\end{proof}

\begin{lemma}
\label{VersionIsNotFromFuture}
If \proc{select}$(q)$ is called at a moment $t$ then value assigned to the variable \id{version} is $\preceq \func{realVersion}(q,t')$ where $t'>t$ is
the moment at which the value is actually being assigned to this variable.
\end{lemma}
\begin{proof}
The proof is similar to that for Lemma \ref{VersionIsFresh} except we use Fact \ref{RevisionIsMonotone} to compare everything to the moment $t'$, not $t$.
\end{proof}

\begin{lemma}
\label{InvalidationWorks}
Let $\epsilon$ be an upper bound for the time necessary to execute the \proc{invalidate} method.
Let $q_w,q_r$ be such, that $\func{subspace}(q_w) \cap \func{subspace}(q_r) \neq \emptyset$.
If \proc{invalidate}$(q_w)$ is called at a moment $t$ and
\proc{select}$(q_r)$ is called at the moment $t+\epsilon$, then the variable \id{version} is assigned a value $\succ \func{realVersion}(q_r,t)$. 
\end{lemma}
\begin{proof}
From Definition \ref{EDefinition} and Lemma \ref{TransitiveClosure} we know that there exists $p\in P$ such that $(q_w,p) \in E' \land (p,q_r) \in E''$.
This particular $p$ is important, as it connects \proc{select}$(q_r)$ and \proc{invalidate}$(q_w)$.
At moment $t+\epsilon$ the \proc{invalidate} is already finished, so the algorithm for \proc{invalidate} had to call \call{globalCache}{increment}(p) before $t+\epsilon$.
There are two cases.
\begin{itemize}
\item \call{globalCache}{increment}(p) was successful and assigned a new value strictly larger than \func{realRevision}$(p,t)$, implying that $\func{realRevision}(p,t) < \func{realRevision}(p,t+\epsilon)$ and thus $\func{realVersion}(q_r,t) \prec \func{realVersion}(q_r,t+\epsilon) \preceq \id{version}$.
\item or it failed due to the key $p$ being missing. 
At some moment between $t$ and the assignment to the variable \id{version} a successful \call{globalCache}{add} for the key $p$ had to occur. 
From Lemma \ref{WriteStrictlyRises} we know that the value assigned had to be strictly larger than $\func{realRevision}(p,t)$, and thus again we get $\func{realVersion}(q_r,t) \prec \id{version}$
\end{itemize}
\end{proof}

\begin{theorem}
The result returned by the proposed \proc{select}$(q)$ algorithm when called at a moment $t$ is equal to the result that would be returned by \id{database}.\proc{select}$(q)$ if called at some moment $t' \ge t - \epsilon$ where $\epsilon$ is the upper bound for the time between calling \id{database}.\proc{delete} or \id{database}.\proc{insert} and exiting from \proc{invalidate} method.
\end{theorem}
\begin{proof}
The theorem holds trivially if the algorithm was forced to actually execute \call{database}{select}, 
so let us concentrate on the more interesting cases when the result was served from the \id{localCache} or the \id{globalCache}.

The algorithm verified that $\attribii{cached}{version} \succeq \id{version}$ before returning the \attribii{cached}{result}.
From Lemma \ref{VersionIsFresh} we know, that $\id{version} \succeq \func{realVersion}(q,t)$.
Consider the moment $t''$ when the thread which put \id{cached} in the \id{globalCache} performed \call{database}{select}$(q)$.
By Lemma \ref{VersionIsNotFromFuture} we know that $\attribii{cached}{version} \preceq \func{realVersion}(q,t'')$.
Using transitivity of $\succ$ relation we get $\func{realVersion}(q,t) \preceq \func{realVersion}(q,t'')$.
If $t'' \ge t - \epsilon$, then let $t' = t''$ and we are done.
Otherwise there are two cases to consider:
\begin{itemize}
\item there was a call to \call{database}{delete} or \call{database}{insert} during the period $(t'',t-\epsilon)$ which resulted in a call to \proc{invalidate}$(q')$ such that $(q',q)\in E$. From Lemma \ref{InvalidationWorks} we immediately get $\func{realVersion}(q,t) \succ \func{realVersion}(q,t'')$ which is a contradiction,
\item otherwise in the period $(t'',t-\epsilon)$ there was no write to the database that could affect \call{database}{select}$(q)$. A hypothetical result of calling it at moment $t'=t-\epsilon$ would be equal to the \attribii{cached}{result}.
\end{itemize}
\end{proof}

%% file: optimizations.tex
\section{Optimizations}
\label{Optimizations}

Instead of using a \id{localCache} to store cached results, we can use the \id{globalCache} only and fetch all required information in a single \proc{multiget}.
This may be a good choice if communication latency is a more important issue than bandwidth, as we trade a single local \proc{get} for a larger global \proc{multiget} here. 

An important optimization is to limit the number of dimensions to only those which are ever used in queries with non-* values.
For example if a table \id{User} contains 20 columns, but equality constraints in \kw{WHERE} clauses of statements involve only \id{Id}, \id{Email} or \id{Password}, then
there is no point in using full 20-dimensional model of the query space. It is enough to project the space onto the 3 important dimensions only, which greatly reduces the amount of revision keys involved.

A more sophisticated optimization is to trim the dependency graph even further by observing what type of queries the system performs.
It is rarely the case that all possible subspaces are really generated by read and write queries.
Those nodes that do not correspond to any known query can be removed, together with all edges and all nodes in the middle layer which became isolated.
Intuitively a revision counter is only really needed if it is both incremented and read, otherwise we can remove it from the graph.
For the example mentioned in Section \ref{Introduction}, write queries always have subspace of the form $(U,G,D)$ or $(U,*,*)$ while subspace of a read query always has a form $(*,G,D)$.
Using this domain specific knowledge one can tune the algorithm so that \SELECT $(*,G,D)$ will depend only on $(*,G,D)$ and $(*,?,?)$, while every \INSERT $(U,G,D)$ will invalidate only $(*,G,D)$ and every \DELETE $(U,*,*)$ will invalidate only $(*,?,?)$. 
This way \SELECT query will need to fetch only 2 revision keys, while \INSERT and \DELETE will have to increment only a single key.
\begin{figure}[!htb]
\includegraphics[width=\columnwidth]{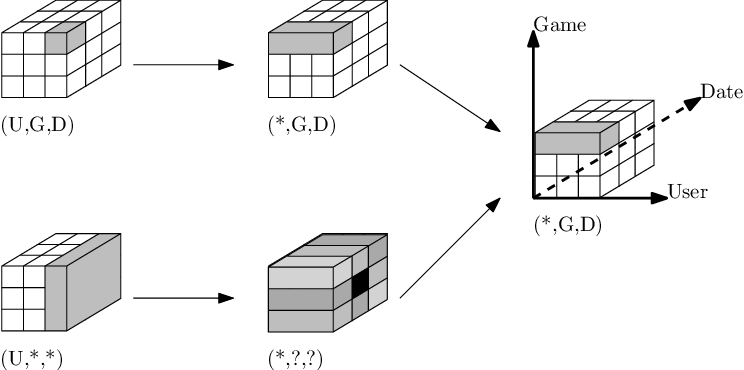}
\label{Trimmed}
\end{figure}

Note that this is exactly the same procedure, as the manual solution mentioned in Section \ref{Introduction}, but with more systematic names for the main revision numbers and revision number for each particular pair $(G,D)$. Our algorithm refers to them as $(*,?,?)$ and $(*,G,D)$ respectively.
An illustration of the fragment of a trimmed graph is presented in Figure \ref{Trimmed}. 
It is just a fragment, as nodes parametrized by $U$,$G$, or $D$ should have multiple copies in the graph, one for each possible value of parameters. We can see that it is trimmed though, as the picture contains only one (not 8) outgoing edges from each query.
In general this can be greatly automated by providing a white list of patterns for write and read queries, which can be then used to deduce a minimal set of middle layer nodes which need to be incremented, as well as to test if system generates only white-listed queries.

Some systems which perform Create, Read, Update, Delete (CRUD) operations, actually create rows one at a time, and also update them one-by-one.
This behavior can be seen in many RESTful applications, where CRUD operations are mapped to HTTP verbs POST, GET, PUT, DELETE.
Even if the system deletes multiple rows at once, one could often emulate it with a loop deleting one row after another, charging the cost of the overhead for each row directly to the \INSERT statement which created the row. 
Therefore it is quite realistic to consider systems in which write queries always have 0-dimensional subspaces, that is subspaces without stars. 
Observe that in such systems we will never have to increment a revisions for nodes containing a question mark, as the only rule which results in incrementing them requires a star in the original query. 
As noted before, this allows us to trim the graph and remove all the nodes from the middle layer which contain a question mark. 
Moreover nodes in the right layer are by definition connected only to nodes in the middle layer which differ only at positions with question marks. Since we now have no question marks at all, it implies that the node in the right layer is now connected only to a single node in the middle layer, the one which has exactly equal label.
In other words, for such CRUD systems we need to fetch only a single revision number from the cache during each \kw{SELECT}. 
The algorithm becomes much simpler and faster.

Some cache implementations allow to bulk increment operations into a single packet. 
This could be used to optimize invalidation into a single round trip to the \id{globalCache}.

Simplistic model presented in Section \ref{Model} required us to handle \kw{OR} inside the \kw{WHERE} clause with extra caution, leading to a safe overestimation of the subspace scanned by a query. For example\\ 
\mbox{\kw{WHERE} (\id{User} = 2 \SQLOR \id{User} = 2) \SQLAND \id{Game}=7}\\ 
becomes upper bounded to \kw{WHERE} \id{Game}=7, and thus is subject to invalidation more often than necessary.
The framework can be easily adapted to handle \kw{OR} more efficiently.
Assume the \kw{WHERE} clause is in a DNF. For each clause compute the subspace separately.
If this is a \proc{delete} query, then call \proc{invalidate} for each subspace separately.
If this is a \proc{select} query, then \id{variants} should be computed as a concatenation of \proc{allVariantsOf} computed for each subspace.

In Section \ref{Model} we explained that inequalities are not handled optimally by the algorithm -- they are simply discarded.
For small domains of integers we could use the following emulation, which uses binary representation.

A column storing $w$-bit integers can be virtually replaced with $w$ columns storing each bit separately.
A range constraint \mbox{\id{x} \kw{BETWEEN} :A \SQLAND :B} can be rewritten so that it uses only $O(w)$ clauses with equality constraints.
For example \id{x} \kw{BETWEEN} 1 \SQLAND 7 becomes
$x_2 = 0$ \SQLAND $x_1 = 0$ \SQLAND $x_0 = 1$ \SQLOR  
$x_2 = 0$ \SQLAND $x_1 = 1$ \SQLOR 
$x_2 = 1$ 
which corresponds to $O(w)$ subspaces : $(0,0,1)$, $(0,1,*)$, and $(1,*,*)$. 
In general degree of each of these subspaces in the revision dependency graph is $O(2^w)$, which may be prohibitive.

Observe, however, that these subspaces have a specific form -- stars in a tuple always form a suffix.
The substitution rules used during read queries will thus never generate a question mark to the right of a star.
Also, the rules used for write queries will never put a question mark to the left of a number.
Therefore, by directly applying the trimming idea presented before, we can restrict the middle layer only to nodes in which the tuple consists of three separate (possibly empty) parts: a prefix of numbers, a suffix of stars and questions marks in the middle.
It can be shown, that this optimization alone reduces the degree of nodes in left and right layer to $O(w)$, and the total number of keys accessed during a query to $O(w^2)$. 

But we can do better.
Notice, that nodes in the middle layer with exactly the same numeric prefix are always incremented together and thus (at least if we ignore evictions) should have always the same value. This leads to another optimization : replacing multiple counters which share the same value with a single one.
Let us replace a counters $(1,?,*)$,$(1,*,*)$ and $(1,?,?)$ with a single $(1,\%,\%)$, etc.
The intuition behind this is that tuples without percent signs correspond to subtrees of a full binary tree spanned over the integers, which together cover the range of the query, while tuples with percent signs correspond to their ancestors. 
As each range can be covered by $O(w)$ subtrees, which have many common ancestors, the total number of accessed keys can easily be shown to be $O(w)$.
In other words a single $w$-bit column increases the number of keys accessed during each query $O(w)$ times and does not increase the number of communication round trips.

For simplicity and minimal technical requirements this paper considers storing revision keys in a volatile cache. This comes at the cost of complicated procedures handling cache misses. In practice it would be wiser to store them in a in-memory database backed up by an append log stored to a permanent memory, such as Redis. This introduces another problem though -- lifespan of a revision key becomes infinite and at some point we can run out of memory if number of combinations of parameters used in queries is not limited.

In some scenarios it may be an important optimization to invalidate cache if and only if write operation actually affected any rows.

%% file: experiment.tex
\section{Performance tests}
The proposed algorithm has been used in applications having more than half a million users for several months 
without any problems. 
Additionally, for the purpose of this paper the algorithm was tested in the following artificial setting.
There were 10 threads written in PHP, communicating with two different Memcached servers (local and global) and with a MySQL database containing a simple 3-dimensional table.
Each thread performed 10 000 random operations from the list:
\begin{itemize}
\item {\INSERT} a random point of the 10x10x10 integer grid,
\item {\DELETE} a random one-dimensional line, 
\item {\SELECT} a random two-dimensional plane.
\end{itemize}
Test results for various probabilities of these operations are shown in Table \ref{Results}.
Before each test caches where empty and the table was filled with 500 equally spaced points.
\begin{description}
\item[cache hits] is the number of times a result from current version was found in cache and \call{database}{select} was not performed
\item[stale] is the number of times the algorithm returned a stale value found in a cache due to the latency of the invalidation algorithm (the $\epsilon$)
\item[max age] is the age of the most stale result ever returned from cache (the $\epsilon$). Of course this value greatly depends on the performance of the machine, so it is given here just to illustrate how large $\epsilon$ is in practice  
\item[med age] is the median of ages of stale results returned from cache. Fresh results are excluded here
\item[fresh] is the number of times the value returned from cache was actually exactly the same as the database would return
\item[naive hits] is the number of hits achieved by a hypothetical naive algorithm which flushes whole cache after each modification of the database
\item[inserts] is the number of inserts not ignored due to the uniqueness constraint
\item[deletes] is the number of deletes which removed at least one record
\end{description}
\begin{tabular}{l|rrrrr}
\label{Results}
ppb of \proc{select}      &  99\%   &  98\%  & 90\%     & 80\%    & $1/3$   \\
ppb of \proc{insert}      &  0.9\%  &   1\%  &  9\%     & 10\%    & $1/3$   \\
ppb of \proc{delete}      &  0.1\%  &   1\%  &  1\%     & 10\%    & $1/3$   \\
\hline \hline
selects                   & 98956   & 97964  & 89952    & 80298   & 33408   \\
\quad cache misses        & 2876    & 9089   & 24110    & 48529   & 31260   \\
\quad cache hits          & 96080   & 88875  & 65842    & 31769   & 2182    \\
\quad hit ratio           & 97\%    & 91\%   & 73\%     & 35\%    & 7\%     \\
\quad \quad stale         & 440     & 1527   & 4858     & 6343    & 963     \\
\quad \quad \quad max age & 0.400s  & 0.304s & 0.288s   & 0.216s  & 0.148s  \\
\quad \quad \quad med age & 0.032s  & 0.032s & 0.040s   & 0.032s  & 0.036s  \\
\quad \quad fresh         & 95640   & 87348  & 60984    & 25426   & 1219    \\
\quad \quad fresh ratio   & 99\%    & 98\%   & 92\%     & 80\%    & 56\%    \\
inserts                   & 498     & 898    & 4822     & 9079    & 30599   \\
deletes                   & 113     & 653    & 953      & 5781    & 19378   \\
\hline \hline
naive hits                & 28687   &33720   & 11555    & 8367    & 741     \\
naive misses              & 70269   &64244   & 78397    & 71931   & 32701   \\
naive hit ratio           & 29\%    &34\%    & 12\%     & 10\%    & 2\%     \\
\end{tabular}